\documentclass[12pt]{article}

\usepackage{latexsym,amssymb,amsfonts,amsmath,amsthm}
\usepackage[dvips]{graphicx}
\usepackage{comment}

\newtheorem{theorem}{Theorem}

\newtheorem{corollary}{Corollary}
\newtheorem{proposition}{Proposition}
\newtheorem{remark}{Remark}
\newtheorem{definition}{Definition}

\numberwithin{theorem}{section}
\numberwithin{lemma}{section}
\numberwithin{corollary}{section}
\numberwithin{proposition}{section}
\numberwithin{remark}{section}
\numberwithin{definition}{section}

\newcommand{\bs}[1]{\boldsymbol{#1}}

\newcommand{\SL}[1]{\stackrel{\circ}{#1}}

\newcommand{\dist}{{\rm dist}}

\usepackage{color}

\title{Characterization of Green's function of discrete Schr\"odinger operator on a finite graph by its spanning subgraphs}
\author{Yusuke Higuchi$^1$, Etsuo Segawa$^2$\\
$^1${\small Department of Mathematics,
Gakushuin University,} \\
{\small Tokyo 171-8588, Japan}
\\
$^2${\small Graduate School of Environment and Information Sciences, Yokohama National University,}\\ {\small Hodogaya, Yokohama 240-8501, Japan
}
}
\date{}

\begin{document}

\maketitle

\par\noindent
{\bf Abstract}. 
The Green's function of the discrete Sch\"odinger operator on a finite graph is considered. This setting reproduces Laplacian and signless Laplacian by adjusting appropriate potentials. We show two ways of the expression for the Green's function by using graph structures.  
The first way is based on the factor of the graph by subtrees which have uni-self-loops; the second way is based on that by odd unicycle graphs. 
\\

%abst
%\footnote[0]{

%}

\section{Introduction}
In the research area of the discrete spectral geometry,
one of main topics is to characterize the spectral structure of a discrete
Sch{\"o}dinger operator
in terms of a certain family of geometric properties of the graph.
When we try to analyze the spectral structure, the Green's function of the operator plays an important role. Once focusing on the Green's function itself,
it is discussed that, for a special class of operators, that is, the Laplacian and the signless Laplacian, the Green's function can be given by the quantity of the geometry of graphs.
For instance, it is well known that the Green's function of the Laplacian of a graph with a grounding point can be expressed by counting the number of spanning trees, and spanning forests constructed by two subtrees (see \cite{Boll} and its references therein). 
On the other hand, the Green's function of the signless Laplacian of a graph can be expressed  by counting the number of spanning odd-unicycle graphs~\cite{CRS2007}. 
%In this paper, we generalize the Laplacian and signless Laplacian, and find the underlying graph structures to express the Green's function by considering the following weighted adjacency matrix with potential and boundary. Let $\mathcal{M}=M+V$ be the discrete Schr{\"o}dinger operator constructed by the weighted adjacency matrix $M$ with the potential $V$. 
In this paper, we give the geometric structure of the underlying graph to express the Green's function for a general class of discrete Schr{\"o}dinger operators including the Laplacian and the signless Laplacian.
Let us briefly state our setting; more general setting and results will be discussed in subsequent sections. A graph $\Gamma = (X,A)$ is a finite graph with the sets of vertices $X$ and edges $A$;
$\delta X \subset X$ is fixed as the boundary of $\Gamma$. Moreover, for the weighted adjacency operator $M$ with a potential
function $V$ on $\Gamma$, we set the discrete Schr\"odinger operator  $\mathcal{M}_V = M + V$.
In this paper we consider the following problem:  
\[  (\mathcal{M}_V-zI) f= g\]
with the boundary condition such that for any $x\in \delta X$, $f(x)=g(x)=0$. 
The solution of the above problem can be described as
\[ f(x)=\sum_{y\in X\setminus \delta X}G_z(x,y)g(y), \]
where $G_z(\cdot,\cdot)$ is called the Green's function.
%$G_z\in \mathbb{C}^{(X\setminus \delta X)\times (X\setminus \delta X)}$ such that 
We characterize the Green's function using some graph factors. 
Note that if we set $V(u)=-\deg(u)$, then the Laplacian is reproduced while if we set $V(u)=+\deg(u)$, then the signless Laplacian is reproduced. Recently, in a typical series of studies on quantum walks (for examples, \cite{FelHil1, FelHil2,HS,HS_PON4}, and see \cite{Portugal} for the review on quantum walks and its references therein), which are related to a stationary discrete Schr{\"o}dinger equation~\cite{KHiguchi}, the following generalized Laplacian interpolating the Laplacian ($\alpha=1$) and the signless Laplacian ($\alpha=-1$) is introduced in \cite{HS_PON4}:  
\[L_\alpha=M-j_+(\alpha)D+j_-(\alpha)\Pi_{\delta X}\;\;(\alpha\in\mathbb{C}\setminus\{0\}), \] 
where $D$ is the degree matrix, $\Pi_{\delta X}$ is the projection onto $\mathbb{C}^{\delta X}$ and $j_{\pm}(\alpha)=(\alpha\pm \alpha^{-1})/2$.  
It is shown in \cite{HS_PON4} that its Green's function plays key roles to express the stationary state and also its scattering information of the quantum walk. 
It is also shown that the Green's function of a gram matrix induced by $2$-cell embedding on the orientable surface describes the stationary state of a quantum walk~\cite{HS_OpticalQW}.  
In this paper, our setting includes the above matrices and shows the underlying graph factors.  

This paper is organized as follows. 
In Section~2, the setting of graph and the discrete-the Green's function of the Schor{\"o}dinger equation is presented. In Section~3, we give our main theorem which describe two kinds of the representation for the Green's function using the graph factors and discuss the relation between the two expressions. In Section~4, the proof of  the main theorem is presented. 
In Section~5, we demonstrate the computational way proposed by this paper of the Green funciton using the graph factors. 
\section{Setting}
\subsection{Graph notation}
Let $\Gamma=(X, A)$ be a finite, connected, and symmetric digraph. Here for any $a\in A$, there uniquely exists an inverse arc $\bar{a}$ in $A$. This graph $\Gamma$ may have a multiple edge but no self-loops.  
The terminus and origin vertices of $a\in A$ are denoted by $t(a)$ and $o(a)$, respectively. 
The boundary of $\Gamma$, $\delta X$, is a proper subset of $X$. 
The support edge of $a\in A$, which is undirected, is described by $|a|$. Note that $|a|=|\bar{a}|$. We set $E=E(\Gamma)=\{|a| \;:\; a\in A\}$ as the set of edges.  
The spanning subgraph $\Gamma'$ (or the factor) of $\Gamma$ is the subgraph of the underlying undirected graph $(X(\Gamma),E(\Gamma))$ satisfying $X(\Gamma)=X(\Gamma')$.  
A spanning tree is a tree and spanning subgraph.
 A spanning forest is a spanning subraph such that each component is a tree.
For a subraph $W$ of $\Gamma=(X,A)$ or $(X,E(\Gamma))$, 
we often write $\Gamma$ for the set of vertices of $W$ without confusion.

\subsection{Adjacency matrix with potential}\label{sect:graphs}
For any finite countable set $\Omega$, we set $\mathbb{C}^{\Omega}$ as the vector space whose standard basis vectors are labeled by $\Omega$. 
The weight on each edge is denoted by a map $w: E\to \mathbb{C}$.  
The weighted adjacency matrix, Laplacian and signless Laplacian on the graph $\Gamma$ are denoted by 
\begin{align*} 
(Mf)(u) &=\sum_{e\in A\;;\;t(e)=u}w(|e|)f(o(e)); \\
(Lf)(u) &= \sum_{e\in A\;;\;t(e)=u}w(|e|)(\; f(t(e))-f(o(e)) \;); \\
(Qf)(u) &= \sum_{e\in A\;;\;t(a)=u}w(|e|)(\; f(t(e))+f(o(e)) \;); 
\end{align*}
for any $f\in \mathbb{C}^X$ and $u\in X$.  
%The matrix representations in $\mathbb{C}^{X\times X}$ are described by
%\begin{align*}
%M[u,v] &= \sum_{t(e)=u,o(e)=v}w(|e|) \\
%L[u,v] &= \begin{cases}
%\sum_{t(e)=u,o(e)=v}w(|e|) & \text{: $u\neq v$} \\
%-\sum_{t(e)=u}w(|e|) & \text{: $u=v$}
%\end{cases} \\
%Q[u,v] &= \begin{cases}\sum_{t(e)=u,o(e)=v}w(|e|) & \text{: $u\neq v$} \\ \sum_{t(e)=u}w(|e|) & \text{: $u= v$}
%\end{cases}
%\end{align*}
We set the weight on $X$ which is a map $V: X\to \mathbb{C}$ as a potential. 
We use the same notation as the potential $V\in \mathbb{C}^X$ for the multiplication operator such that $(Vf)(u)=V(u)f(u)$ for any $f\in \mathbb{C}^{X}$ and $u\in X$. 
Let $\mathcal{M}_V:\mathbb{C}^X\to \mathbb{C}^X$ be the weighted adjacency matrix of $\Gamma$ with the potential $V$ such that 
\[ \mathcal{M}_V=M+V  \]
for any $f\in \mathbb{C}^{A}$ and $u\in X$. 
Note that if $V(u)=-\sum_{t(e)=u} w(|e|)$, then $L=M+V$ while $V(u)=+\sum_{t(e)=u} w(|e|)$, then $Q=M+V$. 
Then we will concentrate on $\mathcal{M}_V$. 

Fix $\delta X \subsetneq X$ and  
let $\chi: \mathbb{C}^{X}\to \mathbb{C}^{X\setminus \delta X}$ such that 
\[ (\chi f)(u)=f(u) \]
for any $f\in \mathbb{C}^{X}$ and $u\in X\setminus \delta X$. Then the adjoint $\chi^{*}: \mathbb{C}^{X\setminus \delta X}\to \mathbb{C}^{X}$ is described by 
\[ (\chi^* g)(u)=\begin{cases} g(u) & \text{: $u\in X\setminus \delta X$,}\\ 0 & \text{: otherwise} \end{cases} \]
for any $g\in \mathbb{C}^{X\setminus \delta X}$ and $u\in X$. 
Note that $\chi$ is represented by the following $|X\setminus \delta X|\times |X|$ matrix  
\[\chi\cong [ I_{X\setminus \delta X} \;|\; \bs{0} ]. \]
Our interest is the resolvent of $\chi \mathcal{M}_V \chi^*$; that is, 
\[ G(z;\Gamma;\delta X;V):=(\chi \mathcal{M}_V\chi^*-zI )^{-1}.  \]
We will describe it using the factors of graphs. 

To this end, let us deform the original graph in the following way. \\

\noindent{\bf The graph deformation $\stackrel{\circ}{\Gamma}=(X,E\cup S_{\delta X})$:}
For the undirected graph $(X(\Gamma),E(\Gamma))$, we add the self-loop to every vertex $X\setminus \delta X$; we write $S_{\delta X}\cong X\setminus \delta X$ for the set of the self loops.  
Such a resultant graph is denoted by $\stackrel{\circ}{\Gamma}$. 

%%%Otameshi↓
For an undrected graph $\Gamma'$, 
let $\mathcal{OUC}(\Gamma')$ be the set of all the connected and odd unicyclic subgraphs in $\Gamma'$, where the odd unicyclic graph is a graph which has exactly one cycle and its length is odd.  Remark that for any subgraph in $\mathcal{OUC}(\Gamma')$, if an appropriate edge is removed from this subgraph, then it becomes a tree. 
We regard each self-loop as an odd cycle and an isolated vertex as a tree in this paper. Then note that an isolated vertex with one self-loop belongs to $\mathcal{OUC}(\SL{\Gamma})$.  
As a special class in $\mathcal{OUC}(\SL{\Gamma})$, we put  $\SL{\mathcal{T}}\subset \mathcal{OUC}(\SL{\Gamma})$ by the subset of $\mathcal{OUC}(\SL{\Gamma})$ each element of which is a tree with one self-loop. 
Set the following families of subgraphs of $\SL{\Gamma}$ which are the ``required graph parts" to describe the resolvent and depend on the operators $L$ and $Q$ as follows:  
\begin{align*}
\mathcal{F}_J(\Gamma; \delta X)&=\begin{cases}
\SL{\mathcal{T}}
& \text{: $J=L$,} \\
\mathcal{OUC}(\SL{\Gamma}) & \text{: $J=Q$,} \end{cases}\\
\mathcal{F}_J(\Gamma; \delta X; \ell,m)&=
\begin{cases}
\SL{\mathcal{T}}(\ell,m) & \text{: $J=L$,} \\
\mathcal{OUC}(\SL{\Gamma};\ell,m) & \text{: $J=Q$} 
\end{cases}
\end{align*}
for $u_\ell,u_m\notin \delta X$. 
Here $\SL{\mathcal{T}}(\ell,m)$ is the subset of $\SL{T}$ such that every subgraph has neither the vertices $u_\ell$ nor $u_m$, 
and $\mathcal{OUC}(\SL{\Gamma};\ell,m)$ is the subset of $\mathcal{OUC}(\SL{\Gamma})$ such that each subgraph has neither the vertices $u_\ell$ nor $u_m$.  Note that $\SL{\mathcal{T}}(\ell,m)\subset \mathcal{OUC}(\SL{\Gamma};\ell,m)$. 
See Section~\ref{sect:example} for its example. 
%%%%
\begin{definition}[The graph factor for describing the resolvent]\label{def:graphfactor}
Let $\Gamma=(X,A)$ be a connected symmetric graph with the vertex set $X=\{u_1,\dots,u_N\}$ and the boundary $\delta X$. 
Then the graph factors of $\SL{\Gamma}=(X,E\cup S_{\delta 
 X})$ are defined as follows. 
For $J\in \{L,Q\}$, 
\begin{align*}
\mathcal{H}_J&(\Gamma;\delta X)  \\
&:=
\{ \text{$H$ is a spanning subgraph of $\SL{\Gamma}$ }\;|\; \\
&\qquad\qquad 
\text{ For each connected component $W$ in $H$, }\\
&\qquad\qquad  
\text{ (i) if $W \cap \delta X \neq \emptyset$, then  
$|W\cap \delta X|=1$ and $W$ is a tree; }\\
&\qquad\qquad 
\text{ (ii) otherwise, $W\in \mathcal{F}_J(\Gamma;\delta X)$}
\}
\end{align*}
on the other hand, for $u_\ell, u_m\notin \delta X$,  
\begin{align*}
\mathcal{H}_J&(\Gamma;\delta X;\ell,m)  \\
&:=\{ \text{$H$ is a spanning subgraph of $\SL{\Gamma}$ }\;|\; \\
&\qquad
\text{ For each connected component $W$ in $H$, }\\
&\qquad\qquad
\text{ \;\;$(i)'$ if $W\cap \delta X\neq \emptyset$, then $|W\cap \delta X|=1$, $W\cap \{u_\ell,u_m\}=\emptyset$ and $W$ is a tree; }\\ 
&\qquad\qquad 
\text{ \;(ii)$'$ if $W\cap \{u_\ell,u_m\}\neq \emptyset$, then $W\supset \{u_\ell,u_m\}$, $W\cap \delta X =\emptyset$ and $W$ is a tree;  }\\
&\qquad\qquad 
\text{ (iii)$'$ otherwise, $W\in \mathcal{F}_J(\Gamma;\delta X;\ell,m)$ }
\}
\end{align*}

\end{definition}
See Section~\ref{sect:example} for its  example. 
%%%%%%%Otameshi↑
%%%%%%%%%%%%%
\section{Main theorem and related results}\label{sect:main}
\subsection{Main theorem}
For a given $w: E\to \mathbb{C}$ and $V: X\to \mathbb{C}$, 
we put the weight $w_J$ on $E\cup S$ ($J=L,Q$) by 
\[ w_L(z;e)=\begin{cases} -w(e) & \text{: $e\in E$,} \\
-z+V_L(e) & \text{: $e\in S_{\delta X}\cong X\setminus \delta X$,}
\end{cases}
\]
\[
w_Q(z;e)=\begin{cases} +w(e) & \text{: $e\in E$,} \\
-z+V_Q(e) & \text{: $e\in S_{\delta X}\cong X\setminus \delta X$,}
\end{cases}
\]
where 
\[V_J(u):=\begin{cases}
V(u)+\pi(u) & \text{: $J=L$,}\\
V(u)-\pi(u) & \text{: $J=Q$,}
\end{cases}
\] 
for any $u\in X\setminus \delta X$. 
Here we set $\pi(u):=\sum_{a\in A\;;\;t(a)=u} w(|a|)$ for any $u\in X$. 
It is easy to see that, for a given $w$ and $V$,  
\[ \mathcal{M}_V=L_{w_L}+V_L=Q_{W_Q}+V_Q=-L+V_L=Q+V_Q, \]
where $L_{w_L}$ and $Q_{w_Q}$ are defined by the weight $w_L$ and $w_Q$., respectively, and $L$ and $Q$ are defined by the weight $w$. 

%%

%%%%%%
\begin{definition}[Weights of spanning subgraph and  families of spanning subgraphs] \label{def:weightgraph} 
\noindent \\
For a subgraph $H\subset \SL{\Gamma}$, the weight of $H$ is defined by  
\[ W_J(z;H)=\prod_{e\in E(H)}w_J(z;e).  \]
The weights of families $\mathcal{H}_J(\Gamma;\delta X)$ and $\mathcal{H}_J(\Gamma;\delta X;\ell,m)$ are defined by 
\begin{align*}
    \iota_1(z;\mathcal{H}_J(\Gamma;\delta X)) &= 
    \begin{cases}
    \sum_{H\in \mathcal{H}_L(\Gamma;\delta X)}W_L(z;H) & \text{: $J=L$}\\
    \\
    \sum_{H\in \mathcal{H}_Q(\Gamma;\delta X)}4^{b_1(H\setminus S(H))}W_Q(z;H) & \text{: $J=Q$}
    \end{cases}
    \\
    \iota_2(z;\mathcal{H}_J(\Gamma;\delta X;\ell,m)) &= 
    \begin{cases}
    \sum_{H\in \mathcal{H}_L(\Gamma;\delta X;\ell,m)}W_L(z;H) & \text{: $J=L$}\\
    \\
    \sum_{H\in \mathcal{H}_Q(\Gamma;\delta X;\ell,m)}4^{b_1(H\setminus S(H))}\;(-1)^{\dist_H(u_\ell,u_m)}\;W_Q(z;H)& \text{: $J=Q$}
    \end{cases}
\end{align*} 
Here $\dist_H(u_\ell,u_m)$ is the shortest length of path between $u_\ell$ and $u_m$ in $H$, $b_1(H\setminus S(H))$ is the number of connected components of $\mathcal{OUC}(\Gamma)$ in $H$ and $S(H)$ is the set of self-loops in $H$. 
\end{definition}

%%%%%%%%%%%%%%%%
\begin{theorem}\label{thm:main}
Let $\Gamma=(X,A)$ be a connected symmetric graph with the vertex set $X=\{u_1,\dots,u_N\}$ and the boundary $\delta X=\{u_{N-|\delta X|+1},\cdots,u_N\}$.  
Then the resolvent of $\mathcal{M}_V$,  $G_z:=G(z;\Gamma;\delta X;V)\in \mathbb{C}^{(X\setminus \delta X)\times (X\setminus \delta X)}$, is formed by the following Green's function
\begin{align}
G_z(u_\ell,u_m) &= \frac{\iota_2(z;\mathcal{H}_L(\Gamma;\delta X;\ell,m))}{\iota_1(z;\mathcal{H}_L(\Gamma;\delta X))} \label{eq:main1}\\
&= \frac{\iota_2(z;\mathcal{H}_Q(\Gamma;\delta X;\ell,m))}{\iota_1(z;\mathcal{H}_Q(\Gamma;\delta X))} \label{eq:main2}
\end{align}
for any $u_\ell,u_m\in X\setminus\delta X$. 
\end{theorem}
%%%%%%%%%%
\begin{remark}\label{rem:MLQ}
If we set $V'(u)=V(u)+\pi(u)$, then 
\[ G(z;\Gamma;\delta X;V)=(\chi\mathcal{L}_{V'}\chi^*-zI)^{-1}, \]
where $\mathcal{L}_{V'}=-L+V'$. \\
On the other hand, if we set 
$V''(u)=V(u)-\pi(u)$, then 
\[ G(z;\Gamma;\delta X;V)=(\chi\mathcal{Q}_{V''}\chi^*-zI)^{-1}, \]
where $\mathcal{Q}_{V''}=Q+V''$. 
\end{remark}
\noindent This means that the resolvent of the weighted Laplacian with an arbitrary potential $V$ can be obtained by replacing $V(u)$ with $V(u)-\pi(u)$ in Theorem~\ref{thm:main} while that of the weighted signless Laplacian with the potential $V$ can be obtained by replacing $V(u)$ with $V(u)+\pi(u)$ in Theorem~\ref{thm:main}. \\
\subsection{Related results}
Before going to the proof of Theorem~\ref{thm:main}, 
let us compare the two expressions (\ref{eq:main1}) and (\ref{eq:main2}) of the same Green's function which derive from $\det(\chi(\mathcal{L}_{V'}-zI)\chi^*)$ and $\det(\chi(\mathcal{Q}_{V''}-zI)\chi^*)$, respectively. 
\begin{proposition}\label{prop:iota12}
\[\iota_1(\mathcal{H}_L(\Gamma;\delta X))=\iota_1(\mathcal{H}_Q(\Gamma;\delta X)),\;\iota_2(\mathcal{H}_L(\Gamma;\delta X;\ell,m))=\iota_2(\mathcal{H}_Q(\Gamma;\delta X;\ell,m))\]
for any $u_\ell,u_m\in X\setminus \delta X$. 
\end{proposition}
\begin{proof}
By Remark~\ref{rem:MLQ}, it holds that  $\Theta_L:=\chi(\mathcal{L}_{V'}-zI)\chi^*$ and 
$\Theta_Q:=\chi(\mathcal{Q}_{V''}-zI)\chi^*$ are equal to $\chi(\mathcal{M}_V-zI)\chi^*=G^{-1}(z;\Gamma;\delta X; V)$.  
Through the proof of Theorem~\ref{thm:main}, it turns out that $\det\Theta_J$ is
identical with  $\iota_1(\mathcal{H}_J(\Gamma;\delta X))$ for each $J\in \{L,Q\}$. Then we have \[\iota_1(\mathcal{H}_L(\Gamma;\delta X))=\iota_1(\mathcal{H}_Q(\Gamma;\delta X)).\]
On the other hand, since the denominators of (\ref{eq:main1}) and $(\ref{eq:main2})$ are the same, that is, $\iota_1(\mathcal{H}_L(\Gamma;\delta X))=\iota_1(\mathcal{H}_Q(\Gamma;\delta X))$, the numerators of (\ref{eq:main1}) and $(\ref{eq:main2})$ are the same; that is, 
\[\iota_2(\mathcal{H}_L;\Gamma;\delta X;\ell,m)=\iota_2(\mathcal{H}_Q;\Gamma;\delta X;\ell,m).\]
\end{proof}
Using this proposition, we give an expression for the number of typical spanning forests. In particular,  (\ref{eq:1}) for $|\delta X|=1$ is known as  the matrix-tree theorem (for example see \cite{Boll}). By Theorem~\ref{thm:main}, some quantities (see LHSs of  (\ref{eq:2}) and (\ref{eq:2'}) ) which are induced by the graph factors $\mathcal{H}_Q(\Gamma;\delta X)$ and $\mathcal{H}_Q(\Gamma;\delta X;\ell,m)$ related to the odd unicylic factor are equivalent to those of spanning forests. 
%%%
\begin{corollary}\label{cor:OUC->tree}
Let $L_{\delta X}:= \chi (M+V)\chi^*$ be the Laplacian on $\kappa$-regular graph, where $w(e)=-1$ and $V(u)=\deg(u)$ with the boundary $\delta X$ ($|\delta X|\geq 1$), (which is the same as $\chi L\chi^*$ with $w(e)=1$.)
Let $N_{\delta X}$, $N_{\delta X;\ell,m}$ be the subsets of the spanning forests such that 
\begin{align*} 
N_{\delta X} &= \{ \text{$H$ is a spanning forest} \;|\; \\
& \qquad\qquad  \text{(i) $\omega(H)=|\delta X|$};\\
& \qquad\qquad \text{(ii) for any connected component $W$ in $H$, $|W \cap \delta X|= 1$}\}, 
\end{align*}
and for $u_\ell,u_m\in X\setminus \delta X$, 
\begin{align*} 
N_{\delta X;\ell,m} &= \{ \text{$H$ is a spanning forest} \;|\; \\
& \qquad\qquad  \text{(i) $\omega(H)=|\delta X|+1$};\\
& \qquad\qquad \text{(ii) for any connected component $W$ in $H$,} \\
& \qquad\qquad\qquad\qquad \text{if $W \cap \delta X\neq \emptyset$, then $|W\cap \delta X|=1$ and $W\cap \{u_\ell,u_m\}=\emptyset$, }\\
& \qquad\qquad\qquad\qquad \text{otherwise, $W\supset\{u_\ell,u_m\}$} \}, 
\end{align*}
respectively. 
Then we have 
\begin{align}
|N_{\delta X}| &= \det L_{\delta X} \label{eq:1}\\
&= (-1)^{|X|-|\delta X|} \sum_{H\in \mathcal{H}_Q(\Gamma; \delta X)} 4^{\omega(H)-|S(H)|-|\delta X|}(-2\kappa)^{|S(H)|}. \label{eq:2}\\
|N_{\delta X;\ell,m}| &= (-1)^{\ell+m}\det( (L_{\delta X;{p,q}})_{p\neq \ell,q\neq m} ) \label{eq:1'}\\
&= (-1)^{|X|-|\delta X|-1} \sum_{H\in \mathcal{H}_Q(\Gamma; \delta X; \ell,m)} 4^{\omega(H)-|S(H)|-|\delta X|-1}(-2\kappa)^{|S(H)|}. \label{eq:2'}
\end{align}
\end{corollary}
\begin{proof}
By setting $w(e)=-1$ and $V(u)=\deg(u)$, $z=0$, the Green's function $\chi(\mathcal{M}_V-zI)\chi^*$ is of the form $L_{\delta X}$. 
In such a setting, the weights are 
\[
w_L(e;z):=\begin{cases} 1 & \text{: $e\in E$}\\ 0 & \text{: $e\in S_{\delta X}$} \end{cases},\;\;
w_Q(e;z):= \begin{cases} -1 & \text{: $e\in E$}\\ 2\kappa & \text{: $e\in S_{\delta X}$} \end{cases}
\]
Since the weight at the self-loop is $0$ for $J=L$, it is enough to see only the spanning forests in the subgraph $\mathcal{H}_{L}(\Gamma;\delta X)$, that is, 
\[ \det L_{\delta X}=\sum_{H\in \mathcal{H}_{L}(\Gamma;\delta X)} 1 = \sum_{H\in N_{\delta X}} 1 =|N_{\delta X}|, \]
which implies (\ref{eq:1}). 
Note that the number of edges except the self-loops in $H\in \mathcal{H}_Q(\Gamma;\delta X)$ is $|X|-|\delta X|-|S(H)|$, and the number of odd unicycle graphs in the connected components of $H\in \mathcal{H}_Q(\Gamma;\delta X)$ is $\omega(H)-|\delta X|-S_{\delta X}(H)$, then 
(\ref{eq:main2}) directly implies  
\[ \det L_{\delta X}=\sum_{H\in \mathcal{H}_Q(\Gamma; \delta X)} 4^{\omega(H)-|\delta X|-S_{\delta X}(H)}(2\kappa)^{S_{\delta X}(H)}(-1)^{|X|-|\delta X|-|S_{\delta X}(H)|}, \]
which leads (\ref{eq:2}).
The proof of (\ref{eq:1'}) and (\ref{eq:2'}) can be also done in a similar method. 
\end{proof}
Let us set 
\[\stackrel{\triangle}{\mathcal{T}}(\Gamma;\delta X):=\{ H\in \mathcal{H}_Q(\Gamma;\delta X)\;|\; \text{some component of $H$}\in \mathcal{OUC}(\Gamma) \}. \]
The family of the spanning subgraph  $\mathcal{H}_Q(\Gamma;\delta X)$ can be decomposed into 
\[\mathcal{H}_Q(\Gamma;\delta X)=\mathcal{H}_L(\Gamma;\delta X)\cup \stackrel{\triangle}{\mathcal{T}}(\Gamma;\delta X).\] 
Thus in the computation of (\ref{eq:main2}), we need an ``extra" computation  corresponding to the ``$\stackrel{\triangle}{\mathcal{T}}(\Gamma;\delta X)$" part comparing with (\ref{eq:main1}). We extract such part using the following expression. 
%%%%%%%
\begin{proposition}\label{prop:relation}
It holds that 
\begin{equation}\label{eq:ibu}
\det(\chi\mathcal{M}_V\chi^*-zI)-(-1)^{|X\setminus \delta X|}\det(\chi\mathcal{M}_{-V}\chi^*+zI)=\sum_{H\in \stackrel{\triangle}{\mathcal{T}}(\Gamma;\delta X) } 4^{b_1(H\setminus S(H))}W_Q(H;z). 
\end{equation}
\end{proposition}
%%%%%%%%%%%%%%%%
\begin{proof}
Let us set the weight of the subgraph $H$  with the potential $V$ be $W_J(H;z;V):=W_J(H;z)$ $(J\in \{L,Q\})$. 
If $H\in \mathcal{H}_L(\Gamma;\delta X)$, then 
\begin{align} W_L(H;z;V) 
&= \prod_{e\in E(H)\setminus S(H)} (-w(e))\cdot \prod_{e\in S(H)} (z+V(e)+\pi(e)) \notag\\ 
&= (-1)^{|X\setminus \delta X|} \prod_{e\in E(H)\setminus S(H)} w(e) \cdot \prod_{e\in S(H)} (-z-V(e)-\pi(e)) \notag\\
&=  W_Q(H;-z;-V). \label{eq:relationLQ}
\end{align}
Here we used $|E(H)|=|X\setminus \delta X|$ in the second equality. 
By Proposition~\ref{prop:iota12}, it holds that 
\begin{equation}\label{eq:detexpression}
\det(\chi\mathcal{M}_V\chi^*-zI)=\iota_1(\mathcal{H}_L(\Gamma;\delta X))=\iota_1(\mathcal{H}_Q(\Gamma;\delta X)). 
\end{equation}
Note that the family of the spanning subgraph  $\mathcal{H}_Q(\Gamma;\delta X)$ can be decomposed into 
\[\mathcal{H}_Q(\Gamma;\delta X)=\mathcal{H}_L(\Gamma;\delta X)\cup \stackrel{\triangle}{\mathcal{T}}(\Gamma;\delta X).\] 
Here $\iota_1(\mathcal{H}_L(\Gamma;\delta X))$ is 
\[ \iota_1(\mathcal{H}_L(\Gamma;\delta X))=\sum_{H\in \mathcal{H}_L(\Gamma;\delta X)} W_L(H;z;V) \]
by its definition, while 
\begin{align} \label{eq:iotaQ}
\iota_1(\mathcal{H}_Q(\Gamma;\delta X)) 
&= \sum_{H\in \mathcal{H}_L(\Gamma;\delta X)} W_Q(H;z;V) + \sum_{H\in \stackrel{\triangle}{\mathcal{T}}(\Gamma;\delta X)} 4^{b_1(H\setminus S(H))}W_Q(H;z;V).  
\end{align}
Then (\ref{eq:relationLQ}), (\ref{eq:detexpression}) and (\ref{eq:iotaQ}) imply 
\begin{align*}
\sum_{H\in \stackrel{\triangle}{\mathcal{T}}(\Gamma;\delta X)} 4^{b_1(H\setminus S(H))}W_Q(H;z;V) &= \iota_1(\mathcal{H}_Q(\Gamma;\delta X)) - \sum_{H\in \mathcal{H}_L(\Gamma;\delta X)} W_Q(H;z;V) \\
&= \det (\chi\mathcal{M}_V\chi^*-zI) - (-1)^{|X\setminus \delta X|}\sum_{H\in \mathcal{H}_L(\Gamma;\delta X)}W_L(H;-z;-V) \\
&= \det (\chi\mathcal{M}_V\chi^*-zI) - (-1)^{|X\setminus \delta X|}\det (\chi\mathcal{M}_{-V}\chi^*+zI),
\end{align*}
which is the desired conclusion. 
\end{proof}
Using this proposition, we find that the determinant of the signless Laplacian can be described by not only the odd cyclic factor~\cite{CRS2007} but also the spannnig forests as follows. 
\begin{corollary}\label{cor:tree->OUC}
Assume $\Gamma$ is a $\kappa$-regular graph. 
Let us consider the case for $\delta X=\emptyset$ and 
set $\mathcal{F}(\Gamma)$ as the family of the spanning forest of $\Gamma$ and 
\[\mathcal{OUCF}(\Gamma):=\{H\in \stackrel{\triangle}{\mathcal{T}}(\Gamma;\emptyset)\;|\; \text{every connected component of $H$ belongs to $\mathcal{OUC}(\Gamma)$}\}, \] 
which is the family of the odd unicylic factor of $\Gamma$. 
%If the potential $V\in \mathbb{C}^{X}$ is given by $V(u)=m(u)-z$, 
Then we have 
\begin{equation}\label{eq:acco}
\sum_{H\in \mathcal{OUCF}(\Gamma)} 4^{\omega(H)} = 
(-1)^{|X|} \sum_{H\in \mathcal{F}(\Gamma)}  (-2\kappa)^{\omega(H)}\Lambda(H), 
\end{equation}
where for the connected components of $H$, with $H=H_1\sqcup \cdots \sqcup H_{\omega(H)}$, we define  $\Lambda(H):=\prod_{j=1}^{\omega(H)} |X(H_j)|$. 
\end{corollary}
\begin{remark}
The RHS of (\ref{eq:acco}) takes a non-negative value. 
Note that the LHS is $0$ iff $\mathcal{OUCF}(\Gamma)=\emptyset$, that is, $\Gamma$ is bipartite. 
Then the RHS is $0$ iff $\Gamma$ is bipartite.
\end{remark}
\begin{proof}
Under the condition for $V(u)=\pi(u)$ and $z=0$, $\delta X=\emptyset$ and $w(e)=1$ for any $e\in E(\Gamma)$, the weight on $E(\SL{\Gamma})$ is reduced to 
\[w_L(z;e)=\begin{cases} -1 & \text{: $e\in E$,} \\
2\kappa & \text{: $e\in S_{\delta X}$,}
\end{cases}
\] 
while 
\[w_Q(z;e)=\begin{cases} +1 & \text{: $e\in E$,} \\
0 & \text{: $e\in S_{\delta X}$.}\end{cases} \]
The second term of LHS of (\ref{eq:ibu}) in Proposition~\ref{prop:relation} is vanished because the Laplacian is non-invertible. 
Note that the number of edges except the self-loops in $H\in \mathcal{H}_L(\Gamma;\emptyset)$ is $|X|-\omega (H)$. 
Then by (\ref{eq:detexpression}), RHS of (\ref{eq:ibu}) is expressed by
\begin{align}
\text{RHS of (\ref{eq:ibu})} &= 
\iota_1(\mathcal{H}_L(\Gamma;\delta X))
=\sum_{H\in \mathcal{H}_L(\Gamma;\emptyset)} (-1)^{|X|-\omega(H)} (2\kappa)^{\omega(H)}  \notag \\
&= \sum_{H\in \mathcal{F}(\Gamma)} (-1)^{|X|-\omega(H)} (2\kappa)^{\omega(H)} \Lambda(H). \label{eq:ibu2}
\end{align}
Moreover if $H\in \stackrel{\triangle}{\mathcal{T}}(\Gamma;\emptyset)$ has a self-loop, then $W_Q(H;z;V)=0$ by the setting of the potential. Thus the domain of the subgraphs for the summation in RHS of (\ref{eq:ibu}) is reduced to $\mathcal{OUCF}(\Gamma)\subset \stackrel{\triangle}{\mathcal{T}}(\Gamma;\emptyset)$.  Combining it with (\ref{eq:ibu2}), we obtain the desired conclusion.   
\end{proof}

%%
%%%%%%%%%%
\section{Proof of main theorem}
\subsection{Proof of (\ref{eq:main1})}
\begin{proof}
Let $\vec{E}\subset A$ be the set of arcs removing one of the symmetric arcs from every pair of $\{a,\bar{a}\}$.  
The directed graph, $\vec{\Gamma}=(X,\vec{E}\cup S_{\delta X})$, which is isomorphic to $\SL{\Gamma}$, is more suitable as the deformed graph considered here. Thus $|A|/2=|\vec{E}|$. 
Let the set of the self-loops $S_{\delta X}$ be $\{\sigma_u \;|\; u\in X\setminus \delta X\}$. 
The (directed) incident matrix of $\vec{\Gamma}$ is denoted by $B_L: \mathbb{C}^{\vec{E}\cup S}\to \mathbb{C}^{X}$ such that 
\[ (B_L\psi)(u)=\sum_{a\in \vec{E}\cup S\;;\;t(a)=u}\psi(a)-\sum_{a\in \vec{E}\cup S\;;\;o(a)=u}\psi(a) + \psi(\sigma_u) \]
for any $\psi\in \mathbb{C}^{\vec{E}\cup S}$ and $u\in X\setminus \delta X$. Then the adjoint $B_L^*: \mathbb{C}^{X}\to \mathbb{C}^{\vec{E}\cup S}$ is described by 
\[ (B_L^*f)(a)= \begin{cases}
f(t(a))-f(o(a)) & \text{: $a\in \vec{E}$,} \\
f(t(a)) & \text{: $a\in S$.}
\end{cases} \]
Note that the matrix representation of $B_L$ in $\mathbb{C}^{X\times (\vec{E}\cup S)}$ is described as follows: 
\[ (B_L)_{u,a}=\begin{cases}
+1 & \text{: $t(a)=u$, $a\in \vec{E}\cup S$,}\\
-1 & \text{: $o(a)=u$, $a\in \vec{E}$,}\\
0 & \text{: otherwise}
\end{cases}\]
for any $u\in X$, $a\in \vec{E}\cup S$. 
Let $D_{L}: \mathbb{C}^{\vec{E}\cup S}\to \mathbb{C}^{\vec{E}\cup S}$ be the multiplication operator such that 
\[ (D_{L}\psi) (a) = w_L(a)\psi(a)\]
which is the diagonal matrix in the representation in $\mathbb{C}^{(\vec{E}\cup S)\times (\vec{E}\cup S)}$, that is, 
\[ (D_{L})_{a,b}=\begin{cases} 
 w_L(z;a) & \text{: $a=b$,}\\ 0 & \text{: otherwise}\end{cases} \]
for any $a\in \vec{E}\cup S$. 

If we set $V'(u)=V(u)+\sum_{t(a)=u}w(|a|)$, 
the weighted Laplacian with the potential $V'$ is rewritten by 
\[ \mathcal{L}_{V'}-zI=B_LD_LB_L^*,  \]
which is equivalent to $\mathcal{M}_{V}-z I$. See Remark~\ref{rem:MLQ}.
Then we have 
\begin{align*}
G(z;\Gamma;\delta X; V) 
&= (\chi (\mathcal{L}_{V'}-zI)\chi^* )^{-1} \\
&= (\;B_L^{(\delta X)} D_L {B_L^{(\delta X)}}^* \;)^{-1},
\end{align*}
where $B_L^{(\delta X)}:=\chi B_L$. 
Set $\Theta_L:=\chi (\mathcal{L}_{V'}-zI)\chi^*$.
Assume $X\setminus \delta X=\{u_1,\dots,u_{N-q}\}$ ($0\leq q\leq N-1$).  
If $\Theta_L$ has the inverse, then 
\[ (\Theta_L^{-1})_{\ell,m}= (-1)^{\ell+m}\frac{\det \Theta_L^{(m,\ell)}}{\det{\Theta_L}}=(-1)^{\ell+m}\frac{\det \Theta_L^{(\ell,m)}}{\det{\Theta_L}} \]
for any $1\leq \ell,m\leq N-q$ 
since $\Theta_L$ is symmetric. 
Here $\Theta_L^{(\ell,m)}$ is the submatrix of $\Theta_L$ consisting of rows $X\setminus (\delta X\cup \{u_\ell\})$ and columns $X\setminus (\delta X \cup \{u_m\})$.  \\

\noindent{\bf $\star$ Spanning subgraphs induced by $\det \Theta_L$ and $\det \Theta_L^{(\ell,m)}$ }
\begin{enumerate}
\item Spanning subgraph for $\det \Theta_L$: 
By the Binet-Cauchy formula\footnote{Binet-Cauchy formula: for an $m\times n$ matrix $A=[\;\bs{x}_1\;|\;\bs{x}_2\;|\;\cdots|\bs{x_n}\;]$ and an $n\times m$ matrix $B=[\;\bs{y}_1\;|\;\bs{y}_2\;|\;\cdots\;|\;\bs{y}_n\;]^\top$ with $m\leq n$, \[ \det(AB)=\sum_{\{i_1,\dots,i_m\}\subset \{1,\dots,n\} } \det [\;\bs{x}_{i_1}\;|\;\bs{x}_{i_2}\;|\;\cdots\;|\;\bs{x}_{i_m}\;]\; \det [\;\bs{y}_{i_1}\;|\;\bs{y}_{i_2}\;|\;\cdots\;|\;\bs{y}_{i_m}\;]. \]}, 
we have 
\begin{align}
\det \Theta_L 
&= \det (B_L^{(\delta X)}D_L{B_L^{(\delta X)}}^*) = \sum_{H\subset \vec{E}\cup S,\; |H|=N-q}W_L(z;H)\det (B_L^{(\delta X)}(H))^2 \\
&= \sum_{H\subset \vec{\Gamma},\; |A(H)|=N-q}W_L(z;H)\det (B_L^{(\delta X)}(H))^2, \label{eq:L_N-q}
\end{align}
where $B_L^{(\delta X)}(H)$ is the $(N-q)\times (N-q)$ submatrix of $B_L^{(\delta X)}$ consisting of the columns corresponding to $H\subset \vec{E}\cup S$. 
Here we identify a subset $H$ of the arc set $A(\vec{\Gamma})=\vec{E}\cup S$ with the subgraph $H$ of $\vec{\Gamma}$. 
Remark that since there exists no row vectors in $B_L^{(\delta X)}$ corresponding to the vertices of $\delta X$, 
let us regard the vertices of $\delta X$ in the subgraph $H$ as the ``sinks".  

Now let us see that the following situations of the subgraph $H$ with $|A(H)|=N-q$ brings $\det B_L^{(\delta X)}(H)=0$. 
\begin{enumerate}
\item{\it ``The underlying undirected graph of $H$ contains a cycle whose length is greater than $2$"} $\Rightarrow \det(B_L^{(\delta X)})=0$. 
\begin{proof}
Assume $H$ contains a cycle whose length is greater than $2$. 
The arcs whose support edegs form a cycle in $H$ are denoted by $e_0,\dots,e_{s-1}\in \vec{E}$. 
Set $e_0',\dots,e_{s-1}'\in A$ so that each terminal vertex of an arc identifies with the origin vertex of the next arc; that is, 
\[ e_0'=e_0,\text{ and }
e_{j}'=\begin{cases} e_{j} & \text{: $o(e_{j})=o(e_{j-1})$} \\
\bar{e}_{j} & \text{: $o(e_j)\neq t(e_{j-1})$}
 \end{cases} \text{ for $j=1,\dots,s-1$},\]
where the subscript is modulus of $s$. 
The columns of $B_L^{(\delta X)}(H)$ corresponding to the cycle is denoted by  $\bs{b}_1\dots,\bs{b}_s$. Note that 
\[ \bs{b}_j(u)=\begin{cases} 1 & \text{: $t(a_j)=u$}\\ -1 & \text{: $o(a_j)=u$}\\ 0 & \text{: otherwise.} \end{cases}\]
We can see that $\bs{b}_j$' s are linearly dependent since 
\[ \sum_{j=1}^s(-1)^{\delta_{e_j, e_j'}}\bs{b}_j=0.\]
Then $\det B^{(\delta X)}_L(H)=0$. 
\end{proof}
\item{\it ``$H$ contains a connected component which has more than one element of $S\cup \delta X$" $\Rightarrow \det B^{(\delta X)}_L(H)=0$}
\begin{proof}
Assume the connected component which has more than one element of $u_\sigma\in S$ and $u_*\in \delta X$. There is a path connecting between $u_\sigma$ and $u_*$ in the support graph $H$ since $u_\sigma$ and $u_*$ are included in the same connected component. Set the arcs in $\vec{E}$ whose support edges form the path between $u_\sigma$ and $u_*$ by $e_1,\dots,e_s$. The sequence $(e_1',\dots,e_s')$ are defined by 
\[ e_1'=\begin{cases} e_1 & \text{: $o(e_1)=u_\sigma$,}\\ \bar{e}_1 & \text{: $t(e_1)=u_\sigma$.} \end{cases} \]
\[e_{j}'=
\begin{cases} e_{j} & \text{: $o(e_{j})=t(e_{j-1})$} \\
\bar{e}_{j} & \text{: $o(e_j)\neq t(e_{j-1})$}
 \end{cases} \]
The columns of $B_L^{(\delta X)}(H)$ corresponding the path and the self-loop is denoted by $\bs{b}_1,\dots,\bs{b}_s$, and $\bs{b}_0$, respectively. 
Note that 
\[\bs{b}_j(u)=
\begin{cases}
1 & \text{: $t(e_j)=u$, $0\leq j\leq s$, $u\neq u_*$} \\
-1 & \text{: $o(e_j)=u$, $1\leq j\leq s$, $u\neq u_*$} \\
0 & \text{: otherwise.}
\end{cases}\]
for any $u\in X\setminus \delta X$. 
Then $\bs{b}_j$'s are linearly dependent since
\[ \bs{b}_0+(-1)^{\delta_{e_1,e_1'}}\bs{b}_1+\cdots+(-1)^{\delta_{e_s,e_s'}}\bs{b}_s=0. \]
In a similar way, we can construct the linearly dependent column vectors for the rest of the cases; that is,  ``$u_\sigma\in S$, $u_*\in S$" ,  ``$u_\sigma\in \delta X$, $u_*\in \delta X$" cases. 
As a consequence we have $\det B_L^{(\delta X)}(H)=0$ if $H$ contains a connected component which has more than one element of $S\cup \delta X$. 
\end{proof}
\end{enumerate}
The observation (a) implies that
for a factor $H$ with $\det(H)\neq 0$, every connected component $W$ of $H$ must be a tree or a tree with self loops. Moreover by the observation (b), if $W\cap \delta X\neq \emptyset$, then $W$ must be a tree, $|W\cap \delta X|=1$; 
otherwise, that is, if $W\cap \delta X =\emptyset$, then $W$ is a tree or a tree with with exactly one-self-loop. 
On the other hand, the factor considered here must have $N-q-1$ arcs by (\ref{eq:L_N-q}). 
Remark that $|E(T)|=|V(T)|-1$ for every tree $T$, and $|E(C)|=|V(C)|$ for every graph having one and only one cycle $C$. Then if $W\cap \delta X=\emptyset$, $W$ must be a tree with exactly one self-loop. This implies that for a factor $H$ with $\det(H)\neq 0$, every connected component $W$ of $H$ satisfies $|W\setminus \delta X|=|A(W)|$. 
As a consequence, we have   
$H\in \mathcal{H}_L(\Gamma;\delta X)$. 
Then it holds that  
\[ \det \Theta_L 
= \sum_{H\in \mathcal{H}_L(\Gamma;\delta X) }W_L(z;H)\det (B_L^{(\delta X)}(H))^2. \]
\item Spanning subgraph for $\det \Theta_L^{(\ell,m)}$: 
By the Binet-Cauchy formula, we have 
\begin{align}
\det \Theta_L^{(\ell,m)} 
&= \det (B_L^{(\delta X\cup\{u_\ell\})}D_L{B_L^{(\delta X\cup\{u_m\})}}^*) \\
&= \sum_{H\subset \vec{E}\cup S,\; |H|=N-q-1}W_L(z;H)\det (B_L^{(\delta X\cup\{u_\ell\})}(H)) \det (B_L^{(\delta X\cup\{u_m\})}(H)) \\
&= \sum_{H\subset \vec{\Gamma},\; |A(H)|=N-q-1}W_L(z;H)\det (B_L^{(\delta X\cup\{u_\ell\})}(H)) \det (B_L^{(\delta X\cup\{u_m\})}(H))
\label{eq:L_N-q-1}
\end{align}
Here $B_L^{(\delta X \cup \{v\})}$ is the submatrix of $B_L^{(\delta X)}$ omitting the row vector corresponding to $v\in X$. 
The following situations of the subgraph $H$ with $|A(H)|=N-q-1$ brings $\det B_L^{(\delta X\cup\{u_k\})}(H)=0$ ($k\in\{\ell,m\}$) by a similar reason to the consideration for the case of $\det(B_L^{(\delta X)}(H))$:
\begin{enumerate}
\item{\it $H$ contains a cycle whose length more than $2$;}
\item{\it $H$ contains a connected component which has more than one element of $S\cup \delta X\cup \{u_k\}$.}
\end{enumerate}
Then by just replacing $\delta X$ with $\delta X \cup \{u_k\}$ in the discussion of (1) on $B^{(\delta X)}_L(H)$, 
this implies that 
for a factor $H$ with $\det(B_L^{\delta X\cup\{u_k\}})\neq 0$, every connected component of $H$, $W$, must be a tree with $|W\cap (\delta X\cup \{u_k\})|=1$ or a tree with exactly one-self-loop with $|W\cap (\delta X\cup \{u_k\})|=0$. 
This implies that 
\begin{equation}\label{eq:A(W)}
|A(W)|=|W\setminus (\delta X \cup \{u_k\})|\text{ $(k=\ell,m)$}. \end{equation}
Next, let us consider the situation of $H$ such that  \[\det B_L^{(\delta X\cup\{u_\ell\})}(H)\det B_L^{(\delta X\cup\{u_m\})}(H)\neq 0.\] 
Assume $u_\ell\neq u_m$. 
For a factor $H$ with $\det B_L^{(\delta X\cup\{u_\ell\})}(H)\neq 0$ and $\det B_L^{(\delta X\cup\{u_m\})}(H)\neq 0$, let us assume that $u_\ell$ and $u_m$ are included in distinct connected components of $H$, say, $W$ and $W'$, respectively. 
Since $|W \cap (\delta X\cup \{u_\ell\})|=1$ and $u_\ell\in W$, $W$ must satisfy $W\cap \delta X=\emptyset$. Then (\ref{eq:A(W)}) implies  
\[ |A(W)|=|W\setminus \{u_\ell\}|=|W\setminus \{u_m\}|. \]
which leads a contradiction to the assumption.   
Then $u_\ell,u_m$ must be included in the same connected component. 
This mean that $H\in \mathcal{H}_L(\Gamma;\delta X;\ell,m)$. 
Then we have 
\[ \det\Theta_L^{(\ell,m)}=\sum_{H\in \mathcal{H}_L(\Gamma;\delta X)}W_L(z;H) {\left(\det B_L^{(\delta X)}(H)\right)}^2.\]
\end{enumerate}
Summarizing the above, we can state that 
\begin{equation}\label{eq:tochuu}
G(z;\Gamma;\delta X;V)_{\ell,m}
=(-1)^{\ell+m}\frac{\sum_{H\in \mathcal{H}_L(\Gamma;\delta X;\ell,m)}W_L(z;H) \det B_L^{(\delta X\cup\{u_\ell\})}(H)\det B_L^{(\delta X\cup\{u_m\})}(H)}{\sum_{H\in \mathcal{H}_L(\Gamma;\delta X)}W_L(z;H) \det^2 B_L^{(\delta X)}(H)}. \end{equation}
Then the next task is the computation of each determinant in RHS of (\ref{eq:tochuu}). \\ 

\noindent {\bf $\star$ Computations of $\det^2(B^{(\delta X)}_L(H))$ and $\det B_L^{(\delta X\cup\{u_\ell\})}(H)\det B_L^{(\delta X\cup\{u_m\})}(H)$}

\begin{enumerate}
\item $\det^2 B_L^{(\delta X)}$: 
The spanning subgraph $H$ consists of  $s$ connected components such that 
\[ H=(\stackrel{\circ}{T}_1\cup \cdots \cup \stackrel{\circ}{T}_{s-q}) \cup (T_{N-q+1}\cup \cdots \cup T_N), \]
where $T_i$ is the tree component of $H$ having $u_i$  ($i=N-q+1,\dots, N$), and $\stackrel{\circ}{T}_j$ is the tree with one self-loop ($j=1,\dots, s-q$). 
Here $s\geq q$ and if $s=q$, then there is not a tree with self-loop in $H$.
Then $B_L^{(\delta X)}(H)$ can be also decomposed into 
\[ B_L^{(\delta X)}(H)\cong \left(B(\stackrel{\circ}{T}_1)\oplus \cdots \oplus B(\stackrel{\circ}{T}_{s-q})\right) \oplus \left(B({T}_{N-q+1})\oplus \cdots \oplus {B}({T}_{N})\right). \]
By taking the cofactor expansions on the row's corresponding to the leaves in each connected component recursively, we obtain 
\begin{equation*} 
\det B(T_i)^2=\det B(\stackrel{\circ}{T_j})^2=1 \;(i=1,\dots,s-q,\;j=N-q+1,\dots, N)  
\end{equation*}
Then we have 
\begin{equation}\label{eq:bunnbo}
\det{}^2 (B_L^{(\delta X)}(H))=1. 
\end{equation}
%%%
\item $\det B_L^{(\delta X\cup\{u_\ell\})}(H)\det B_L^{(\delta X\cup\{u_m\})}(H)$: 
If $u_\ell=u_m$, then by replacing $\delta X$ with $\delta X \cup \{u_\ell\}$ in the discussion of (1), we have 
\[ \det(B_L^{(\delta X\cup \{u_\ell\})}(H))^2=1. \]
So we assume $\ell<m$. 
The subgraph $H$ consists of $s$ connected components such that 
\[ H= T_{\ell,m}\cup (T_{N-q+1}\cup \cdots \cup T_N)\cup (\stackrel{\circ}{T_1}\cup\cdots \stackrel{\circ}{T}_{s-q-1}). \]
Here $T_{\ell,m}$ is the tree connected component having both $u_\ell$ and $u_m$. 
We set $T_{N-q+1}\cup \cdots \cup T_N=: T_{\delta X}$ and $\stackrel{\circ}{T_1}\cup\cdots \cup\stackrel{\circ}{T}_{s-q-1}=: \stackrel{\circ}{T}$. 
In a similar manner to the above case (1), we have 
\[ |\det B_L^{(\delta X\cup\{u_\ell\})}(H)\det B_L^{(\delta X\cup\{u_m\})}(H)|=1. \]

From now on, we will show that this signature is $(-1)^{\ell+m}$. 
The matrix $B_L^{(\delta X \cup \{u_\xi\})}$ is the submatrix of $B_L^{(\delta X)}(H)$ by the deletion of the $\xi$-th row ($\xi\in \{\ell,m\}$).
By some appropriate permutations $\Sigma$ on the columns of $B_L^{(\delta X)}(H)$ and some permutations 
$\sigma_\ell$ and $\sigma_m$ on the rows of $B_L^{(\delta X\cup\{u_\ell\})}(H)$ and $B_L^{(\delta X\cup\{u_m\})}(H)$, respectively, we can express $B_L^{(\delta X \cup \{u_\xi\})}(H)$ as  
\begin{equation}\label{eq:bunkatsu} 
%B_L^{(\delta X\cup \{u_\xi\})}(H) = 
\begin{bmatrix}
B(T_{\ell,m}\setminus \{u_\xi\}) &  &   \\  & B(T_{\delta X}) & \\ & & B(\stackrel{\circ}{T})
\end{bmatrix},
\;\;(\xi\in \{ \ell,m \})
 \end{equation}
where $B(T_{\ell,m}\setminus\{u_\xi\})$ is the submatrix of the incidence matrix of $T_{\ell,m}$ consisting of columns labeled by $E(T_{\ell,m})$ and rows labeled by  $V(T_{\ell,m})\setminus (\delta X \cup \{u_\xi\})$; 
$B(T_{\delta X})$ is the submatrix of the incidence matrix of $T_{\delta X}$ consisting of columns labeled by $E(T_{\delta X})$ and the rows labeled by $V(T_{\delta X})\setminus \delta X$; 
$B(\SL{T})$ is the submatrix of the incidence matrix of $\SL{T}$ consisting of columns labeled by $E(\SL{T})$ and the rows labeled by $V(\SL{T})$.  
Here we may assume the two permutations $\sigma_\ell$ and $\sigma_m$ are expressed as 
\begin{align*}
\sigma_\ell &= \sigma' (1,2)(2,3)\cdots (m-2,m-1), \\
\sigma_m &= \sigma' (1,2)(2,3)\cdots (\ell-1,\ell),
\end{align*}
where $(i,j)$ is the transposition between $i$ and $j$, and $\sigma'$ is an appropriate permutation on $\{2,\dots,N-q-1\}$. 
We should remark that the $(m-1)$-th row of $B_L^{(\delta X \cup \{u_\ell\})}$ corresponds to $u_m$, while the $\ell$-th row of $B_L^{(\delta X \cup \{u_m\})}$ corresponds to $u_\ell$.
In addition, $\sigma_\ell^{-1}(1)$ and $\sigma_m^{-1}(1)$ correspond to the vertices $u_m$ and $u_\ell$, respectively, for $2\leq j\leq N-q-1$, 
$\sigma_\ell^{-1}(j)$ and $\sigma_\ell^{-1}(j)$ correspond to the same vertex of $H\setminus\{u_m,u_\ell\}$. 
Then $\mathrm{sgn}(\sigma_\ell)=(-1)^{m-\ell-1}\mathrm{sgn}(\sigma_m)$.
Using these, we have  
\begin{align} 
\det B_L^{(\delta X\cup\{u_\ell\})}(H)\;&\det B_L^{(\delta X\cup\{u_m\})}(H) \notag \\
&=\mathrm{sgn}(\Sigma)\;\mathrm{sgn}(\sigma_\ell)\;\det B(T_{\ell,m}\setminus \{u_\ell\})\det B(T_{\delta X}) \det B(\SL{T}) \notag \\ 
&\qquad\qquad \times \mathrm{sgn}(\Sigma)\;\mathrm{sgn}(\sigma_m)\;\det B(T_{\ell,m}\setminus \{u_m\}) \; \det B(T_{\delta X}) \det B(\SL{T}) \notag \\
&=(-1)^{m-\ell-1} \det B(T_{\ell,m}\setminus\{u_\ell\})\; \det B(T_{\ell,m}\setminus\{u_m\}). \label{eq:BmBl0}
\end{align}
Put $\tau=|E(T_{\ell,m})|$ which is the matrix size of  $B_{\ell}:=B(T_{\ell,m}\setminus\{u_\ell\})$ and $B_{m}:=B(T_{\ell,m}\setminus\{u_m\})$. 
Let $\bs{b}_1^*$ be the $1$-st row vector of $B_\ell$ and $\bs{b_j}^*$ be the $(j-1)$-th row vectors of $B_m$ for $j=2,\dots,\tau+1$. 
Then 
\[ [\;\bs{b}_1\;|\;\bs{b}_2\;|\;\cdots\;|\;\bs{b}_{\tau+1}\;]^* \]
is just the oriented incidence matrix for $T_{\ell,m}$. 
Then $\bs{b}_1^*+\cdots+\bs{b}_{\tau+1}^*=0$. 
Thus we have 
\begin{align}
\det B_m &= \det[\;\bs{b}_1\;|\;\bs{b}_3\;|\;\cdots\;|\;\bs{b}_{\tau+1}\;]^* \notag \\
&= \det[\;-(\bs{b}_2+\cdots+\bs{b}_{\tau+1})\;|\;\bs{b}_3\;|\;\cdots\;|\;\bs{b}_{\tau+1}\;] \notag \\
&= -\det[\;\bs{b}_2\;|\;\bs{b}_3\;|\;\cdots\;|\;\bs{b}_{\tau+1}  \;] \notag \\
&= -\det B_\ell \label{eq:BmBl}
\end{align}
By combining (\ref{eq:BmBl0}) with (\ref{eq:BmBl}), the final formula is reduced to  the following quite simple form:
\begin{equation}\label{eq:bunshi}
\det B_L^{(\delta X\cup\{u_\ell\})}(H)\det B_L^{(\delta X\cup\{u_m\})}(H)=(-1)^{\ell+m}  \end{equation}
for any $H\in \mathcal{H}_L(\Gamma;\delta X;\ell,m)$. 
\end{enumerate}
Inserting (\ref{eq:bunnbo}) and (\ref{eq:bunshi}) into (\ref{eq:tochuu}), we reach to 
\[
G(z;\Gamma;\delta X;V)_{\ell,m}
=\frac{\sum_{H\in \mathcal{H}_L(\Gamma;\delta X;\ell,m)}W_L(z;H) }{\sum_{H\in \mathcal{H}_L(\Gamma;\delta X)}W_L(z;H) }, \]
which is the desired conclusion. 
\end{proof}
%%%%
\subsection{Proof of (\ref{eq:main2})}
\begin{proof}
Instead of $\Gamma$, we consider the new graph which has the new self-loop to every vertex of $X\setminus \delta X$. Such graph is denoted by $\stackrel{\circ}{\Gamma}=(X,\stackrel{\circ}{E})$ with $ \stackrel{\circ}{E}=E\cup S_{\delta X}$, where $S_{\delta X}\cong X\setminus\delta X$ is the set of the new self-loops. 

Let $B_Q: \mathbb{C}^{\stackrel{\circ}{E}}\to \mathbb{C}^{X}$ be the edge matrix such that 
\[ (B_Q\psi)(u)=\sum_{u\in X\;:\; u\in e}\psi(e), \]
where $u\in e"$ is an end vertex of an undirected edge $e\in \SL{E}$. 
The adjoint $B_Q^*: \mathbb{C}^{X}\to \mathbb{C}^{\stackrel{\circ}{E}}$ is described by 
\[ (B_Q^*\phi)(e)=\begin{cases}
\phi(u)+\phi(v) & \text{: $e=\{u,v\}\in E$}\\
\phi(u) & \text{: $e=\{u\}\in S$}
\end{cases} \]
The diagonal matrix $D_Q$ on $\mathbb{C}^{\stackrel{\circ}{E}}$ is defined by
\[ (D_Q)_{a,b}=\begin{cases}
w_Q(z;a) & \text{: $a=b$} \\
0 & \text{: otherwise}
\end{cases} \]
Then we have 
\[ Q= B^{(p)}\mathcal{K}{B^{(p)}}^{*}.  \]
If we set $V''(u)=V(u)-\sum_{a\in A\;|\; t(a)=u}w(|a|)$, 
then the weighted signless Laplacian with the potential $V''$ is rewritten by 
\[ \mathcal{Q}_{V''}-zI = B_QD_QB_Q^{*} \]
which is equivalent to $\mathcal{M}_V-zI$. See Remark~\ref{rem:MLQ}. Then we have 
\[ G(z;\Gamma;\;\delta X;\;V) = (\chi(\mathcal{Q}_{V''}-zI)\chi_*)^{-1} =(B_Q^{(\delta X)} D_Q {B_Q^{(\delta X)}}^* )^{-1},\]
where $B_Q^{(\delta X)}:= \delta B_Q$. Set $\Theta_Q:=\chi 
(\mathcal{Q}_{V''}-zI) \chi^*$ and assume $X\setminus \delta X=\{u_1,\dots,u_q\}$ $(0\leq q\leq N-1)$. 
If $\Theta_Q$ has the inverse, then 
\[ (\Theta^{-1}_Q)_{\ell,m}= (-1)^{\ell+m}\frac{\det \Theta_Q^{(m,\ell)}}{\det{\Theta_Q}}=(-1)^{\ell+m}\frac{\det \Theta_Q^{(\ell,m)}}{\det{\Theta_Q}} \]
for any $1\leq \ell,m\leq N-q$ 
since $\Theta_Q$ is symmetric. 
Here $\Theta^{(\ell,m)}_Q$ is the submatrix of $\Theta_Q$ consisting of rows $X\setminus (\delta X\cup \{u_\ell\})$ and columns $X\setminus (\delta X \cup \{u_m\})$.  \\

\noindent{\bf $\star$ Spanning subgraphs induced by $\det \Theta_Q$ and $\det(\Theta^{(\ell,m)}_Q)$ }
\begin{enumerate}
    \item Spanning subgraph for $\det \Theta_Q$: 
    By the Binet-Cauchy formula, we have
    \begin{align}
        \det \Theta_Q &= \det (B^{(\delta X)}_Q D_Q {B^{(\delta X)}_Q}^*) 
        = \sum_{ \scriptsize{H\subset E\cup S,\;|H|=N-q }} W_Q(z;H)\det (B^{(\delta X)}_Q(H))^2 \\
        &= \sum_{ \scriptsize{H\subset \stackrel{\circ}{\Gamma},\;|E(H)|=N-q }} W_Q(z;H)\det (B^{(\delta X)}_Q(H))^2, \label{eq:Q_N-q}
    \end{align}
    where $B^{(\delta X)}(H)$ is the $(N-q)\times (N-q)$ submatrix of $B_Q^{(\delta X)}$ consisting of the columns $H\subset \stackrel{\circ}{E}=\cup S_{\delta X} $ of $B^{(\delta X)}$. 
    Here we identify a subset $H$ of the edge set with the subgraph $H$ of $\stackrel{\circ}{\Gamma}$. The vertices of $\delta X$ in the subgraph $H$ is regarded as the sinks. 
    
    Let us see that the following situations of the subgraph $H$ brings  $\det(B^{(\delta X)}_Q(H))=0$. 
    \begin{enumerate}
        \item {\it ``$H$ contains an even cycle"$\Rightarrow \det(B^{(\delta X)}_Q(H))=0$}.  
        \begin{proof}
        Let us denote the even cycle by the sequence of edges in $H$, $(e_1,e_2,\dots,e_{2m})$, and also 
        denote the columns of $B^{(\delta X)}_Q(H)$ corresponding to the even cycle by $\bs{c}_{e_1},\dots,\bs{c}_{e_{2s}}$. Note that $\bs{c}_{e_j}(u)=1$ if $u\in e_{j}$ and $\bs{c}_{e_j}(u)=0$ if $u\notin e_{j}$.  
        Then 
        \[\sum_{j=1}^{2s}(-1)^j\bs{c}_{e_j}=0.\] 
        This means $\bs{c}_{e_j}$'s are linearly dependent which implies $\det B_Q^{(\delta X)}(H)=0$.
        \end{proof}  
        \item{\it $H$ contains a connected component which has  more than one  element of the following set: $\{\text{odd cycles of $\stackrel{\circ}{\Gamma}$}\} \cup \delta X$. Here we regard the self-loop as an odd cycle of length one, then $\det B_Q^{(\delta X)}(H)=0$. } 
        \begin{proof}
        Let us consider the cases that an connected component includes at least one of the following subgraphs: 
        (i) odd cycle--odd cycle; 
        (ii) odd cycle -- $u_*$;
        (iii) $u_*$--$v_*$. 
        Here $u_*,v_*\in \delta X$. 
        Let us consider a typical example for case (i). Assume that the connected component has two odd cycle, $(e_1,\dots,e_{2s+1})$, and another odd cycle, $(f_1,\dots,f_{2\tau+1})$ with $s,\tau>1$, which are connected by a path, $(p_1,\dots,p_r)$ in $H$. Suppose there is a path between them, and $e_1\cap p_1,\;p_r\cap f_1 \neq \emptyset$. Then 
        \[ \sum_{j=1}^{2s+1}(-1)^j\bs{c}_{e_j}-2\sum_{j=1}^{r}(-1)^{j}\bs{c}_{p_j}+(-1)^r\sum_{j=1}^{2\tau+1}(-1)^j\bs{c}_{f_j}=0. \]
        This implies $\det(B_H^{(p)})=0$.
        The other cases can be proven in a similar fashion. 
        \end{proof}
    \end{enumerate}
    Then the above observations (a) and (b) imply that  
    by extending the self-loops $S$ to odd cycles in the discussion of Section~4.1 on $B_L^{(\delta X)}$,  for a factor $H$ with $\det B_Q^{(\delta X)}(H)\neq 0$, every connected component $W$ of $H$ must be a tree with $|W\cap \delta X |=1$ or an odd unicyclic graph with $|W\cap \delta X |=0$. 
    This leads $H\in \mathcal{H}_Q(\Gamma;\delta X)$. 
    Then we have 
    \begin{equation}\label{eq:detQ1}
        \det \Theta_Q= \sum_{H\in \mathcal{H}_Q(\Gamma;\delta X)} W_Q(z;H) \det (B_Q^{(\delta X)}(H) )^2. 
    \end{equation}
    
    %%%
    
%%%%
    \item Spanning subgraph for $\det(\Theta^{(\ell,m)}_Q)$: 
     By the Binet-Cauchy formula, we have
    \begin{align}
        \det \Theta_Q^{(\ell,m)} &= \det (B_Q^{(\delta X\cup\{u_\ell\})}V_Q {B_Q^{(\delta X\cup \{u_m\})}}^*) \\
        &= \sum_{ \scriptsize{H\subset \stackrel{\circ}{E},\;|H|=N-q-1 }} W_Q(z;H)\det B^{(\delta X\cup\{u_\ell\})}_Q(H) \det B^{(\delta X\cup\{u_m\})}(H)) \\
        &= \sum_{ \scriptsize{H\subset \stackrel{\circ}{\Gamma},\;|E(H)|=N-q-1 }} W_Q(z;H)\det B^{(\delta X\cup\{u_\ell\})}_Q(H) \det B^{(\delta X\cup\{u_m\})}_Q(H) \label{eq:Q_N-q-1}
    \end{align}
    where $B_Q^{(\delta X \cup\{u_k\})}$ ($k\in\{\ell,m\}$) is the $(N-q-1)\times |\stackrel{\circ}{E}|$ submatrix of $B^{(\delta X)}$ consisting of rows labeled by vertices in $X\setminus (\delta X\cup \{u_\ell\})$,   
    and $B^{(\delta X \cup \{u_k\})}(H))$ is the $(N-q-1)\times (N-q-1)$ submatrix of $B^{(\delta X\cup\{u_k\})}$ consisting of the columns $H\subset \stackrel{\circ}{E}$. 
    Note that we can regard the subedges set of $H$ as the subgraph of $\SL{\Gamma}$. 
    
    The following situations of the subgraph $H$ brings  $\det(B^{(\delta X \cup\{u_k\})}_{E(H)})=0$ $(k\in\{\ell,m\})$ by a similar reason to the consideration for the case of $\det(B^{(\delta X)}_{E(H)})=0$:
    \begin{enumerate}
    \item {\it $H$ has an even cycle;} 
    \item {\it $H$ contains a connected component which has  more than one element of the following set: $\{\text{odd cycles of $\stackrel{\circ}{\Gamma}$}\} \cup \{u_k\}\cup \delta X $.} 
    \end{enumerate}
\end{enumerate}
Then replacing $\delta X$ with $\delta X \cup \{u_k\}$ in the above discussion of (1) on $\det B_Q^{(\delta X)}$, we conclude that for a factor $H$ with $\det(B^{(\delta X\cup\{u_k\})}_{Q}(H))\neq 0$, every connected component $W$ of $H$ must be a tree with 
$|W\cap (\delta X \cup \{u_k\})|=1$ or an odd unicyclic graph. 
Next, let us consider the situation of $H$ such that  $\det(B^{(\delta X\cup \{u_\ell\})}_Q(H))\det(B^{(\delta X \cup \{u_m\})}_Q(H))\neq 0$. 
By extending self-loop $S$ to odd unicycles in the discussion Section~4.1 on the situation of $H$ satisfying $\det(B^{(\delta X\cup \{u_\ell\})}_L(H))\det(B^{(\delta X \cup \{u_m\})}_L(H))\neq 0$, it is easily to see that 
$u_\ell$ and $u_m$ must be included in the same tree.  Such a factor $H$ belongs to  $\mathcal{H}_Q(\Gamma;\delta X;\ell,m)$. 
Then we have 
\begin{equation}\label{eq:detQ2}
    \det \Theta^{(\ell,m)}_Q= \sum_{H\in \mathcal{H}_Q(\Gamma;\delta X;\ell,m)} W_Q(z;H) \det B^{(\delta X\cup \{u_\ell\})}_Q(H) \det B^{(\delta X\cup \{u_m\})}_Q(H). 
\end{equation} 

\noindent {\bf $\star$ Computations of $\det(B^{(\delta X)}_Q(H))^2$ and $\det(B^{(\delta X \cup\{u_\ell\})}_Q(H))\det(B^{(\delta X\cup\{u_m\})}_Q(H))$}
\begin{enumerate}
    \item $\det(B^{(\delta X)}_{E(H)})^2$:
    The spanning subgraph $H$ is decomposed into 
    \[ H= (\stackrel{\triangle}{T}_1 \cup \cdots \cup  \stackrel{\triangle}{T}_{s}) \cup (\stackrel{\circ}{T}_1 \cup \cdots \cup  \stackrel{\circ}{T}_{s'}) \cup (T_{N-q+1}\cup\cdots\cup T_N),  \]
    where $T_j$ $(j=N-q+1,\dots,N)$ is a tree including exactly one vertex $u_j$ in $\delta X$, $\stackrel{\triangle}{T}_j$ is an odd unicyclic graph and $\stackrel{\circ}{T}_j$ is a tree with one self-loop. 
    Then $B_Q^{(\delta X)}(H)$ can be also decomposed into 
    \begin{multline*} B_Q^{(\delta X)}(H)\cong \left( B(\stackrel{\triangle}{T}_1) \oplus \cdots \oplus  B(\stackrel{\triangle}{T}_s)\right) \oplus \left(B(\stackrel{\circ}{T}_1) \oplus \cdots \oplus  B(\stackrel{\circ}{T}_{s'})\right)
    \\ \oplus
    \left( B(T_{N-q+1})\oplus \cdots \oplus B(T_N) \right)  
    \end{multline*}
    by an appropriate permutation. 
    By taking the cofactor expansions on the rows corresponding to the leaves, recursively, we have
    \[ \det (B(T_i))^2=\det(B(\stackrel{\circ}{T}_j))^2=1.\;(i=N-q+1,\dots,N,\;j=1,\dots,s') \]
    On the other hand, since the determinant of the edge matrix of the odd cycle except the self-loop is $\pm 2$, then by taking the cofactor expansions on the rows corresponding to the leaves, recursively, we have
    \[ \det(B(\stackrel{\triangle}{T}_j))^2=4\;\;\;(j=1,\dots,s). \]
    Then by (\ref{eq:detQ1}), we have 
    \begin{equation}\label{eq:detQ1f} 
    \det \Theta_Q= \sum_{H\in \mathcal{H}_Q(\Gamma;\delta X)}  4^{b_1(H\setminus S(H))}W_Q(z;H)=:\iota_1(z;\mathcal{H}_Q(\Gamma;\delta X)). 
    \end{equation}
   %%%%%%%%%%%%%%%%%%%%%%%%
    \item $\det(B^{(\delta X \cup \{u_\ell\})}_Q(H))\det(B^{(\delta X\cup\{u_m\})}_Q(H))$:
If $u_\ell=u_m$, then by replacing $\delta X$ with $\delta X \cup \{u_\ell\}$ in the discussion of (1), we have 
\[ \det(B_Q^{(\delta X\cup \{u_\ell\})}(H))^2=4^{b_1(H\setminus S(H))}. \]
Now we only have to discuss the case for $\ell<m$ as in the case in Sect.4.1.     
The subgraph is decomposed into 
\[ H= T_{\ell,m} \cup \left(T_{N-q+1}\cup \cdots \cup T_{N}\right) \cup (\stackrel{\triangle}{T}_1 \cup \cdots \cup  \stackrel{\triangle}{T}_{s}) \cup (\stackrel{\circ}{T}_1 \cup \cdots \cup  \stackrel{\circ}{T}_{s'}) \]
Here $T_{\ell,m}$ is the tree including $u_\ell$ and $u_m$. 
We set $T_{\delta X}:=T_{N-q+1}\cup \cdots \cup T_{N}$ and $\text{OUCs}:=(\stackrel{\triangle}{T}_1 \cup \cdots \cup  \stackrel{\triangle}{T}_{s}) \cup (\stackrel{\circ}{T}_1 \cup \cdots \cup  \stackrel{\circ}{T}_{s'})$. 
In a similar manner to the case of $\det B^{(\delta X)}_Q(H)$, we have \[|\det B^{(\delta X \cup \{u_\ell\})}_Q(H) \det B^{(\delta X\cup\{u_m\})}_Q(H)|=4^{b_1(H\setminus S(H))}.\] 
The matrix $B_Q^{(\delta X \cup \{u_\xi\})}(H)$ is the submatrix of $B_Q^{(\delta X)}(H)$ by the deletion of the $\xi$-th row $\xi\in \{\ell,m\}$. 
By some permutation $\Sigma$ on the columns of $B_Q^{(\delta X)}(H)$ and some appropriate permutations $\sigma_\ell$ and $\sigma_m$ on rows of $B_Q^{(\delta X \cup \{u_\ell\})}(H)$ and $B_Q^{(\delta X \cup \{u_m\})}(H)$, respectively, we can express $B_Q^{(\delta X \cup \{u_\xi\})}(H)$ as 
\[  
\begin{bmatrix}
B(T_{\ell,m}\setminus\{u_\xi\}) & & \\
 & B(T_{\delta X}) & \\
 & & B(\text{OUCs}) 
\end{bmatrix},\;(\xi\in\{\ell,m\}) \]
respectively. 
 Here $B(T_{\ell,m}\setminus \{u_\xi\})$ is the submatrix of the incidence matrix consisting of columns labeled by $E(T_{\ell,m})$ and rows labeled by $V(T_{\ell,m})\setminus (\delta X \cup \{u_\xi\})$; $B(\mathrm{OUCs})$ is the incidence matrix of $\mathrm{OUCs}$ consisting of columns labeled by $E(\mathrm{OUCs})$ and the rows labeled by $V(\mathrm{OUCs})$. 
As in the case in Sect.~4.1 we may assume that the two permutations $\sigma_\ell$ and $\sigma_m$ are expressed as 
\begin{align*}
\sigma_\ell=\sigma''(1,2)(2,3)\cdots (m-2,m-1), \\
\sigma_m=\sigma''(1,2)(2,3)\cdots (\ell-1,\ell),
\end{align*}
where $\sigma''$ is an appropriate permutation on $\{2,\dots,N-q-1\}$. 
In a similar way to getting (\ref{eq:BmBl0}), we have 
\begin{multline}
    \det B^{(\delta X\cup\{u_\ell\})}_Q(H) \det B^{(\delta X\cup\{u_m\})}_Q(H) \\
    =  (-1)^{m-\ell-1}\det B(T_{\ell,m}\setminus\{u_\ell\})\;\det B(T_{\ell,m}\setminus\{u_m\}). \label{eq:permutation}    
\end{multline}
In the following, let us determine the signature of $\det(B(T_{\ell,m}\setminus\{u_\ell\}))\det(B(T_{\ell,m}\setminus\{u_m\}))$ because its absolute value is $1$. 
Put $B(T_{\ell,m}\setminus\{u_\ell\})=:B_\ell$ and 
$B(T_{\ell,m}\setminus\{u_m\})=:B_m$ in short. 
Let $\bs{b}_1$ be the first column vector of $B_\ell^*$ and $\bs{b}_j$ ($j=2,\dots,|V(T_\ell,m)|$) be the $(j-1)$-th column vector of $B_m^*$. 
Since $T_{\ell,m}$ is a tree, the vertex set of $T_{\ell,m}$ can be decomposed into the partite set $X$ and $Y$ so that every edge connects a vertex in $X$ to a vertex in $Y$. 
The vertex set of $T_{\ell,m}$ is labeled by $i_1,i_2,\dots,i_{|V(T_{\ell,m})|}$ with $i_1=u_\ell$ and $i_2=u_m$. 
We define an oriented incidence matrix $C$ of $T_{\ell,m}$ by 
\[ C^*=[\bs{c}_1 | \bs{c}_2 | \cdots | \bs{c}_{V(|T_{\ell,m}|)} ], \]
where 
\[\bs{c}_j(e)=\begin{cases} 1 & \text{: $i_j\in e$ and $i_j\in X$} \\
-1 & \text{: $i_j\in e$ and $i_j\in Y$}\\
0 & \text{: otherwise}
\end{cases}  \]
for any $e\in E(T_{\ell,m})$. 
Let $C_\xi$ $(\xi\in \{\ell,m\})$ be the submatrix of $C$ consisting of rows $V(T_{\ell,m})\setminus \{u_\xi\}$.  
By using the diagonal matrix $D_\xi$ $(\xi \in \ell,m)$ on $V(T_{\ell,m})\setminus\{u_\xi\}$ such that 
\[ D_\xi(u)=\begin{cases} 1 & \text{: $u\in X$,} \\ -1 & \text{: $u\in Y$,} \end{cases} \]
we have $C_\xi=D_\xi B_\xi$ $(\xi\in\{\ell,m\})$.  Since the matrix $C$ is a zero sum matrix, by the definitions of the permutations $\sigma_\ell$ and $\sigma_m$,   
\[ \det C_m=-\det C_\ell  \]
holds. 
Then we have 
\begin{align*}
    \det B_\ell &= \det(D_\ell)\det(C_\ell) \\ &= (-1)^{|Y|}(-1)^{\bs{1}_Y(u_\ell)} \det(C_\ell) \\
    \det B_m &= \det(D_m)\det(C_m) \\ &= (-1)^{|Y|}(-1)^{\bs{1}_Y(u_m)} \det(C_m) \\
    &= -(-1)^{|Y|}(-1)^{\bs{1}_Y(u_m)} \det (C_\ell).
\end{align*}
Therefore 
\begin{align}
    \det B_\ell \det B_m 
    &= -(-1)^{\bs{1}_Y(u_\ell)+\bs{1}_{Y}(u_m)} \det (C_\ell)^2 \notag\\
    &= -(-1)^{\dist_H(u_\ell,u_m)}. \label{eq:sign}
\end{align}   
Inserting (\ref{eq:sign}) into (\ref{eq:permutation}), we obtain
\[ \det B^{(\delta X\cup\{u_\ell\})}_Q(H)\; \det B^{(\delta X \cup\{u_m\})}_Q (H)
    =  4^{b_1(H\setminus S)} (-1)^{\dist_H(u_\ell,u_m)}\times (-1)^{m-\ell}.  \]
Thus by (\ref{eq:detQ2}), we obtain 
\begin{align}\label{eq:detQ2f}
    \det \Theta_Q^{(\ell,m)}&= \sum_{H\in \mathcal{H}({\Gamma};\delta X;\ell,m)}  4^{b_1(H\setminus S(H))}W_Q(z;H) (-1)^{\dist_H(u_\ell,u_m)}\times (-1)^{m-\ell} \\
    &=(-1)^{\ell+m} \iota_2(\mathcal{H}_Q(\Gamma;\delta X;\ell,m)). 
\end{align}
\end{enumerate}
By (\ref{eq:detQ1f}) and (\ref{eq:detQ2f}), finally we obtain 
\[ (G(z,\Gamma;\delta X;V))_{\ell,m}=(\Theta_Q^{-1})_{\ell,m}=\frac{\iota_2(\mathcal{H}_Q(z; \Gamma;\delta X;\ell,m))}{\iota_1(z; \mathcal{H}_Q(\Gamma;\delta X))} \]
which is the desired conclusion. 
\end{proof}
%%%%
\section{Example}\label{sect:example}
Let us consider the graph $\Gamma=C_3+ P_2$ as an example.   
The vertices are labeled by $V(C_3)=\{1,2,3\}$ and $V(P_2)=\{1,4\}$.
The boundary is set by $\delta X=\{4\}$. 
See Fig.~1 for the graph. 
The families of the subgraphs $\mathcal{F}_J(\Gamma;\delta X)$ and $\mathcal{F}_J(\Gamma;\delta X;\ell,m)$ ($J=L,Q$,\;$\ell,m=1,2,3$) defined in Section~\ref{sect:graphs} are depicted in Fig.~2. 
The families of the spanning subgraphs $\mathcal{H}_J(\Gamma;\delta X)$, $\mathcal{H}_J(\Gamma;\delta X; \ell,m)$ $(\ell,m=1,2,3,\;J=L,Q)$ in Definition~\ref{def:graphfactor} are depicted by Figure~3. 

Let us consider the case for $w(e)=1$ and $V(u)=0$. In that case, the resolvent of $\chi \mathcal{M} \chi^*$ is simply reduced to 
\begin{align*} G(\Gamma;\delta X; V)&=(\chi \mathcal{M} \chi^*-zI)^{-1} = \begin{bmatrix}
-z & 1 & 1 \\ 1 & -z & 1 \\ 1 & 1 & -z
\end{bmatrix}^{-1} \\
& =\frac{-1}{(z-2)(z+1)^2}
\begin{bmatrix}
(z-1)(z+1) & (z+1) & (z+1) \\
(z+1) & (z-1)(z+1) & (z+1) \\
(z+1) & (z+1) & (z-1)(z+1)
\end{bmatrix}\\
&=
\frac{-1}{(z-2)(z+1)}
\begin{bmatrix}
z-1 & 1 & 1 \\ 1 & z-1 & 1 \\ 1 & 1 & z-1
\end{bmatrix}. 
\end{align*}
On the other hand, let us express the Green's function in terms of spanning subgraphs by 
apply Theorem~\ref{thm:main}. 
The third equality in the above may be useful to check the following computations because we will see that 
\[ \iota_1(\Gamma;\delta X)=(-1) (z-2)(z+1)^2 \]
and each $(\ell,m)$ element in the matrix corresponds to  $\iota_2(\Gamma;\delta X;\ell,m)$. 

It is easy to compute that  
\[ w_L(z;e)=
\begin{cases}
-1 & \text{: $e\in E(\Gamma)$} \\
-z+3 & \text{: $e\in S$, $t(e)=1$} \\
-z+2 & \text{: $e\in S$, $t(e)=2,3$}
\end{cases}\]
\[ w_Q(z;e)=
\begin{cases}
+1 & \text{: $e\in E(\Gamma)$} \\
-z-3 & \text{: $e\in S$, $t(e)=1$} \\
-z-2 & \text{: $e\in S$, $t(e)=2,3$}
\end{cases}\]
by the definition of $w_J$'s ($J\in \{L,Q\}$). 
For each $J\in\{L,Q\}$ case, we divide the spanning subgraphs of $\mathcal{H}_J(\Gamma;\delta)$ and also $\mathcal{H}_J(\Gamma;\delta;\ell,m)$ 
by the equivalent relation ``$\stackrel{W_J}{\sim}$" such that $H \stackrel{W_J}{\sim} H'$ iff $W_J(z;H)=W_J(z;H')$. 
The definition of weight of subgraph $H$, $W_J(z;H)$, can be seen in Defintion~\ref{def:weightgraph}. 
See Fig.~\ref{fig:3}:  
the cardinarity of the quotient set $\mathcal{H}_J(\Gamma;\delta X)/\stackrel{W_J}{\sim}$ is $6$ for $J=L$ while that is $7$ for $J=Q$, and  
the number of connected component of $H$, $\omega(H)$, determines this classification. 
On the other hand, for $\mathcal{H}_J(\Gamma;\delta X;\ell,m)$, 
see Fig.~\ref{fig:4}:  
the cardinarities of the quotient sets 
are
$|\mathcal{H}_J(\Gamma;\delta X;1,1)/\stackrel{W_J}{\sim}|=3$, 
$|\mathcal{H}_J(\Gamma;\delta X;1,2)/\stackrel{W_J}{\sim}|=3$, 
$|\mathcal{H}_J(\Gamma;\delta X;2,3)/\stackrel{W_J}{\sim}|=3$, and 
$|\mathcal{H}_J(\Gamma;\delta X;3,3)/\stackrel{W_J}{\sim}|=4$, and 
the number of the connected component of each $H$, $\omega(H)$, determines this classification. 
For example, for $(\ell,m)=(1,2)$, 
the number of subgraphs with the $2$-connected component is $2$, because of the restriction that the vertices $\{1,2\}$ must be included in the same connected component and such a connected component must be a subtree including no boundary vertices by the definition of $\mathcal{H}_J(\Gamma;\delta X;\ell,m)$. 

Summing up all such weights over all graphs in Fig~3 (for $\iota_1(\cdot)$) and om Fig~4 ($\iota_2(\cdot)$), we can compute them and find its consistency. 
See Fig. 3 (for $\iota_1$) and Fig. 4 (for $\iota_2$). 

%%%%
\begin{figure}[hbtp]
    \centering
    \includegraphics[keepaspectratio, width=50mm]{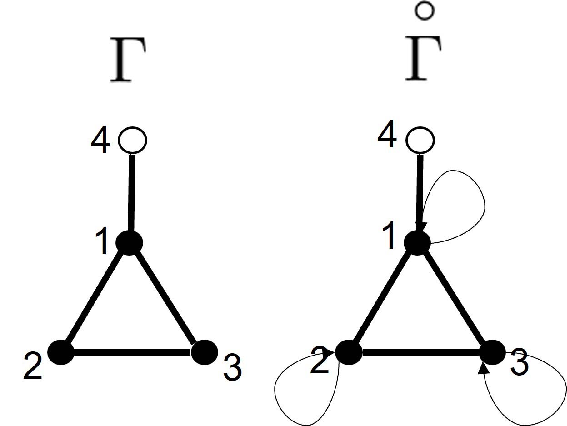}
    \caption{The graph for the example $\Gamma=C_3+ P_2$ with the boundary vertex set $\delta X=\{4\}$. The deformed graph $\SL{\Gamma}$ has self-loop to every vertex of $X\setminus \delta X=\{1,2,3\}$. }
    \label{fig:1}
\end{figure}
\begin{figure}[hbtp]
    \centering
    \includegraphics[keepaspectratio, width=180mm]{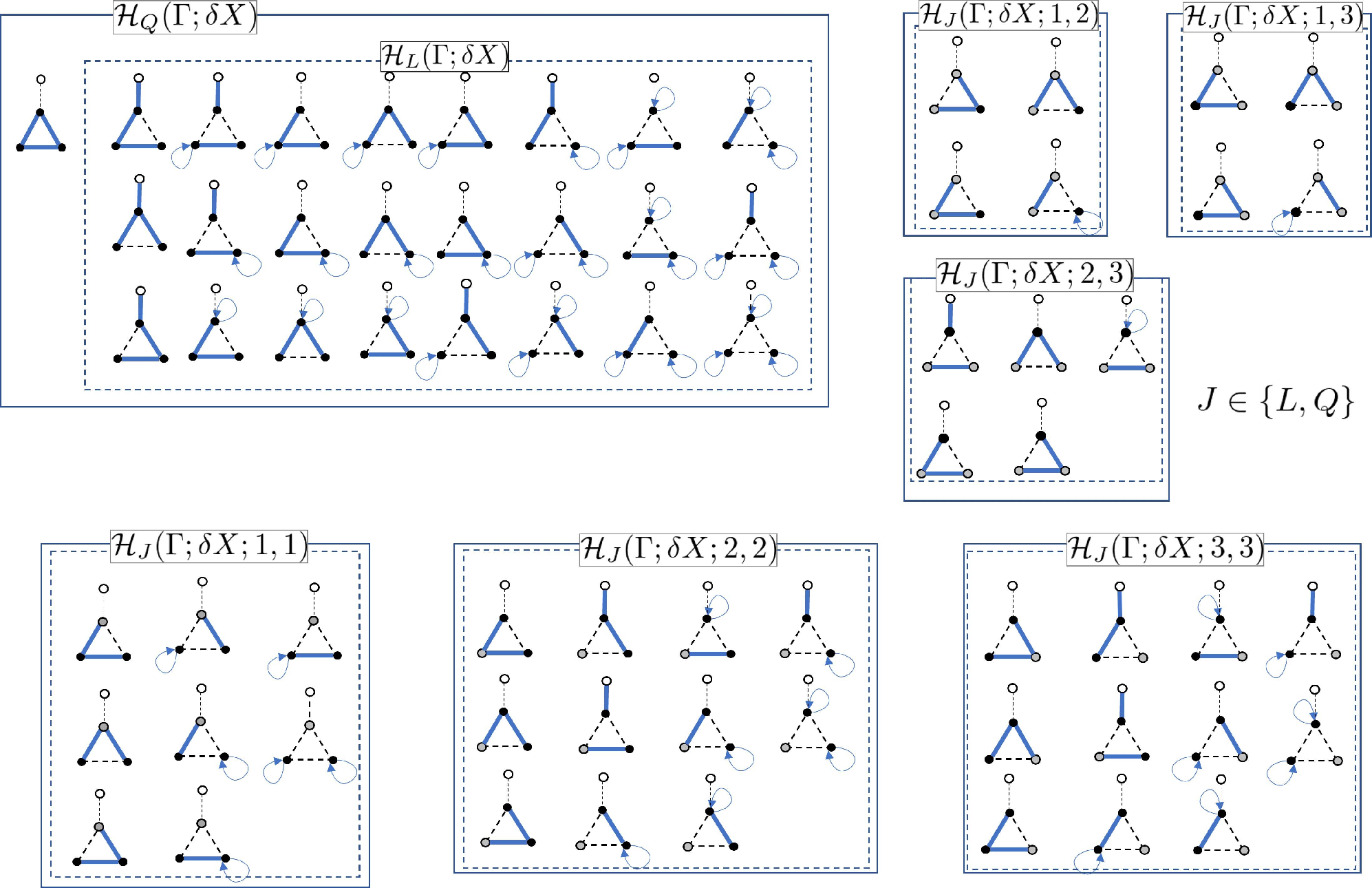}
    \caption{$\mathcal{H}_J(\Gamma;\delta X)$ and $\mathcal{H}_J(\Gamma;\delta X;\ell,m)$:  
    Since $\Gamma$ has only one cycle, whose length is odd, then we have $\mathcal{H}_Q(\Gamma;\delta X)=\mathcal{H}_L(\Gamma;\delta X) \cup \{ (C_3\cup K_1) \}$ and $\mathcal{H}_L(\Gamma;\delta X;\ell,m)=\mathcal{H}_Q(\Gamma;\delta X;\ell,m)$ for any $\ell,m=1,2,3$.    }
    \label{fig:3}
\end{figure}
\begin{figure}[hbtp]
    \centering
    \includegraphics[keepaspectratio, width=180mm]{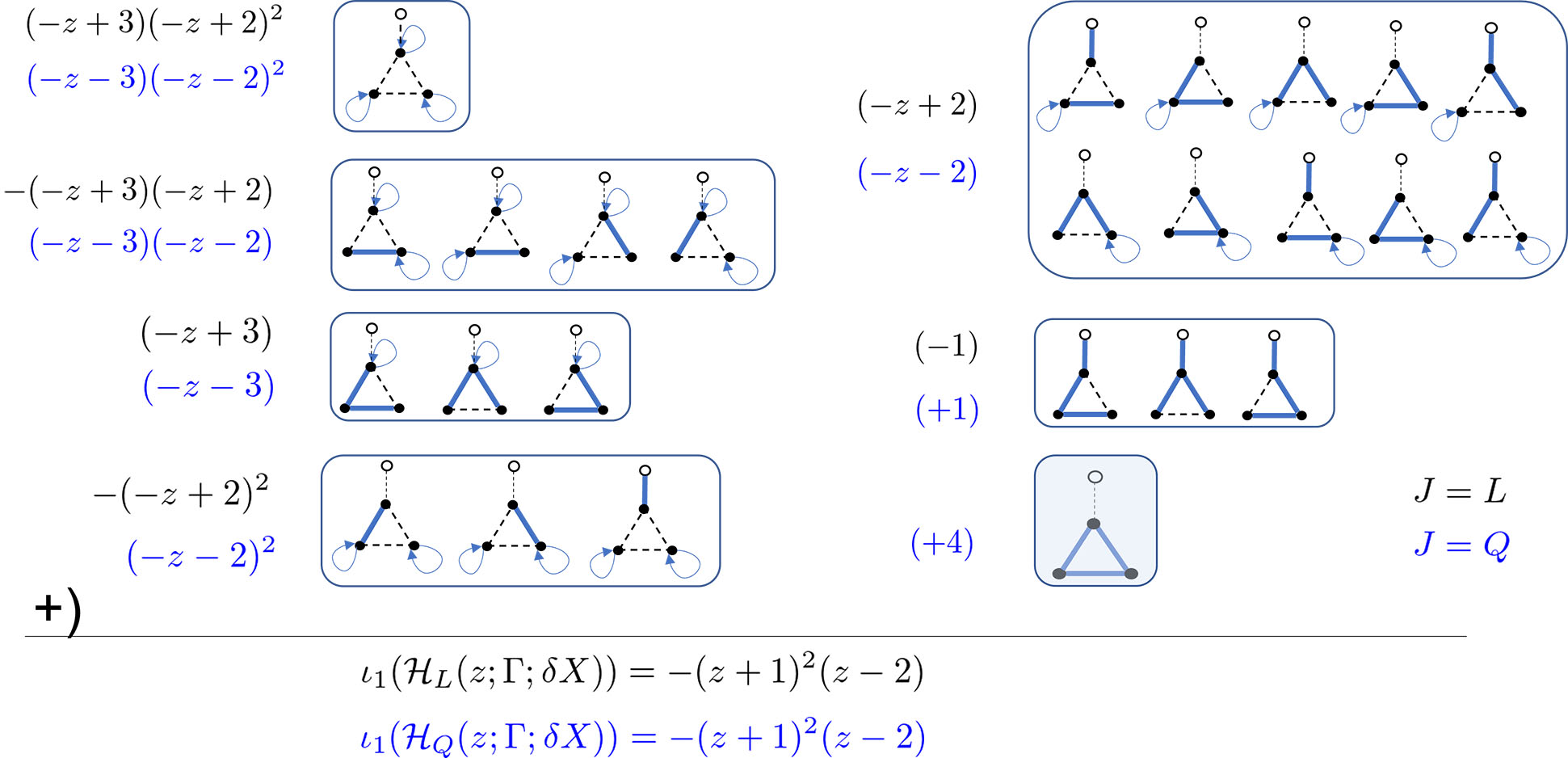}
    \caption{Computation of $\iota_1$ }
    \label{fig:4}
\end{figure}
\begin{figure}[hbtp]
    \centering
    \includegraphics[keepaspectratio, width=180mm]{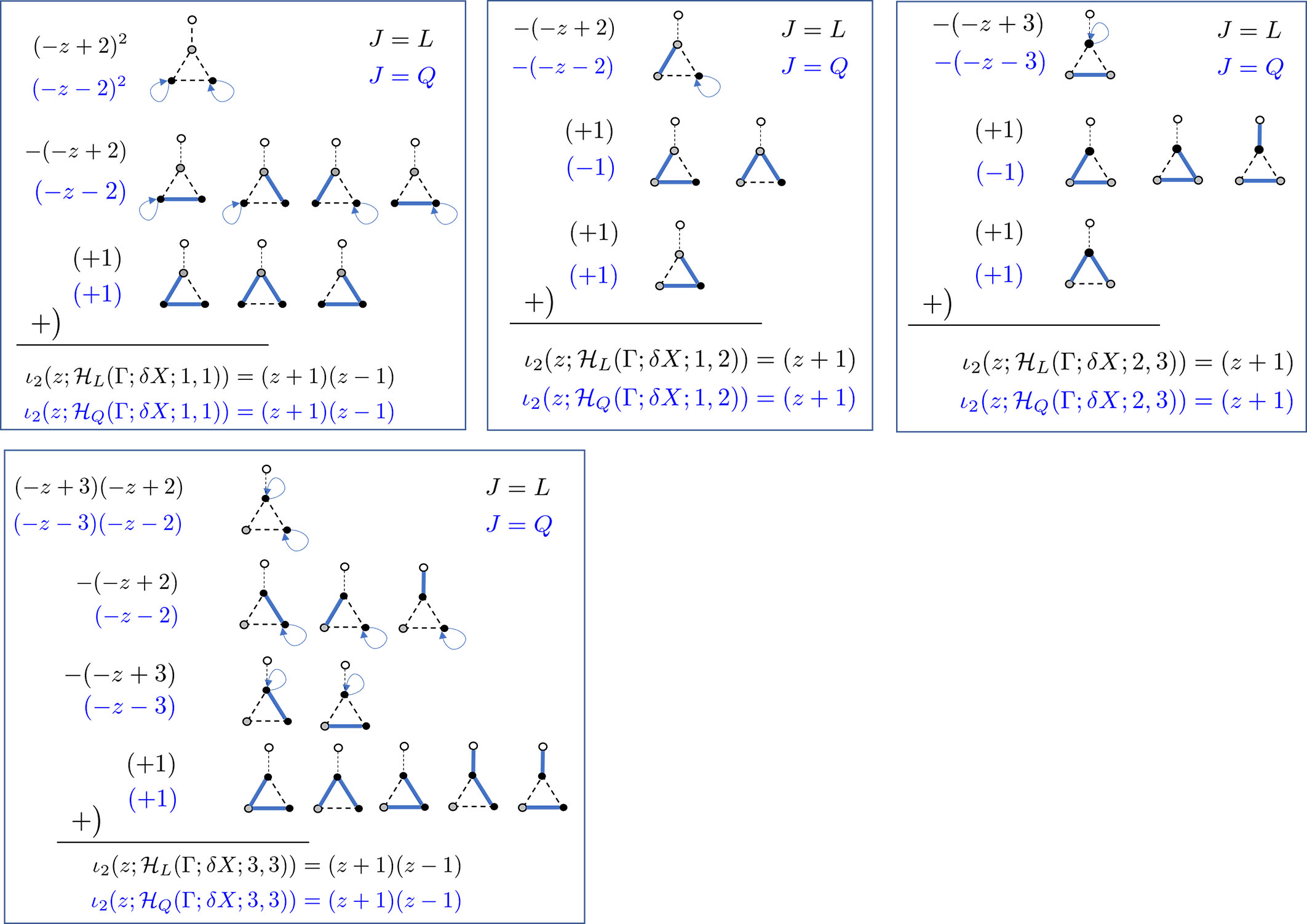}
    \caption{Computation of $\iota_2$ }
    \label{fig:5}
\end{figure}
%%%%%

%\appendix
%\section{Proof of Lemma~\ref{lem:HScor}}
\noindent\\
\noindent {\bf Acknowledgments}
Yu.H. acknowledges financial supports from the Grant-in-Aid of
Scientific Research (C) Japan Society for the Promotion of Science (Grant No.~18K03401). 
E.S. acknowledges financial supports from the Grant-in-Aid of
Scientific Research (C) Japan Society for the Promotion of Science (Grant No.~19K03616) and Research Origin for Dressed Photon.

%\appendix
%\def\thesection{Appendix \Alph{section}}
%\renewcommand{\theequation}{A.\arabic{equation}}
%\setcounter{equation}{0}

%\section{}

\begin{small}
\bibliographystyle{jplain}

\end{small}

\end{document}